\documentclass[journal]{IEEEtran}

\usepackage{hyperref}
\hypersetup{hidelinks}
\usepackage{xcolor,soul,framed} 

\colorlet{shadecolor}{yellow}
\usepackage[pdftex]{graphicx}
\DeclareGraphicsExtensions{.pdf,.jpeg,.png}

\usepackage{amsfonts,amssymb}
\usepackage{mathrsfs}
\usepackage{amsthm}
\usepackage[cmex10]{amsmath}
\usepackage{mathtools}
\usepackage{bbm}
\DeclareMathAlphabet{\mathbbb}{U}{bbold}{m}{n}
\usepackage{array}
\usepackage{flushend}
\usepackage{physics}
\usepackage{cite}
\usepackage{comment}

\theoremstyle{remark}
\newtheorem{theorem}{Theorem}

\newtheorem{proposition}{Proposition}

\newtheorem{remark}{Remark}

\newcommand{\hermconj}  {^{\mathsf{H}}}
\newcommand{\trans}     {^{\mathsf{T}}}
\newcommand{\pha}[1]    {\underline{#1}}
\newcommand{\phaconj}[1]{\overline{\underline{#1}}}
\newcommand{\vect}[1]   {\boldsymbol{#1}}
\newcommand{\phavec}[1] {\underline{\boldsymbol{#1}}}

\newcommand{\mat}[1]    {\boldsymbol{#1}}
\newcommand{\phamat}[1] {\underline{\boldsymbol{#1}}}

\begin{document}
\bstctlcite{IEEEexample:BSTcontrol}
    \title{Passivity and Decentralized Stability Conditions for Grid-Forming Converters}
    \author{Xiuqiang~He,~\IEEEmembership{Member,~IEEE,}
          and~Florian~Dörfler,~\IEEEmembership{Senior Member,~IEEE}
  \thanks{This work was supported by the European Union’s Horizon 2020 research and innovation program under Grant 883985.}
  \thanks{The authors are with the Automatic Control Laboratory, ETH Zurich, 8092 Zurich, Switzerland. Email:\{xiuqhe,dorfler\}@ethz.ch.}}

\maketitle


\begin{abstract}
We prove that the popular grid-forming control, i.e., dispatchable virtual oscillator control (dVOC), also termed complex droop control, exhibits output-feedback passivity in its \textit{large-signal} model, featuring an explicit and physically meaningful passivity index. Using this passivity property, we derive decentralized stability conditions for the \textit{transient stability} of dVOC in multi-converter grid-connected systems, beyond prior small-signal stability results. The decentralized conditions are of practical significance, particularly for ensuring the transient stability of renewable power plants under grid disturbances.
\end{abstract}

\begin{IEEEkeywords}
Complex droop control, dispatchable virtual oscillator control, grid-forming control, passivity, transient stability.
\end{IEEEkeywords}

\vspace{-5mm}

\section{Introduction}

\IEEEPARstart{R}{enewable} energy sources are being ever-increasingly integrated into power grids. \textit{Grid-forming control technologies} are emerging in power systems to maintain system strength and stability. Unlike large-capacity synchronous generators, renewable power generation features a substantial number of small-capacity generating units connected to power grids. Conventional stability analysis faces growing difficulties when applied to investigate numerous renewable power generating units in a centralized manner. These challenges stem from both analytical intractability and computational burdens due to numerous generating units as well as confidentiality restrictions on system-wide data or model availability \cite{liu2022stability}. Additionally, since grid-forming converters operate autonomously, and exhibit rich and varied dynamic behavior, aggregated equivalent modeling and analysis \cite{li2017practical} is rarely applicable.

Decentralized stability analysis has drawn increasing attention in recent years \cite{spanias2019system,dey2023passivity,yang2020distributed}. Decentralized approaches divide nodes and networks, treating each node's dynamics and the network coupling separately, relying solely on local and essential network data. Passivity serves as a fundamental tool in decentralized stability analysis, applicable to both nonlinear and linearized systems \cite{khalil2002nonlinear} (roughly corresponding to transient stability and small-signal stability analysis in power systems). There have been many studies exploring passivity-based conditions for small-signal stability of \textit{linearized} converter systems \cite{spanias2019system,dey2023passivity}. However, the area of decentralized stability analysis for \textit{large-signal transient stability} remains underexplored \cite{liu2022stability}. In \cite{yang2020distributed}, decentralized stability conditions for nonlinear dynamics of synchronous generators and classical grid-forming controls are established, ensuring the \textit{local stability} of a given equilibrium point and corresponding controller setpoints. 

We examine the passivity and decentralized stability conditions for the most recent grid-forming control, dispatchable virtual oscillator control (dVOC), also called complex droop control \cite{he2012quanlitative}. Our results feature \textit{large-signal passivity analysis} for \textit{global asymptotic stability}, which, to the best of our knowledge, cannot or can hardly be achieved through classical grid-forming controls. We reveal that dVOC-controlled multi-converter grid-connected systems, satisfying our prescribed decentralized conditions, achieve large-signal/transient stability.

\section{System Model With Feedback Connection}

Grid-forming converters measure their output currents and establish the terminal voltages, as depicted in Fig.~\ref{fig:block-diagram}(a). In our modeling and analysis, a three-phase balanced condition is considered, and the \textit{grid synchronous reference frame} is employed \cite{spanias2019system} since the system's steady state converges to the grid frequency. In the grid reference frame, the converter \textit{voltage} and the \textit{output current} are expressed with $dq$-axis components as $\pha v_k \coloneqq v_{{\rm d},k} +j v_{{\rm q},k}$ and $\pha i_k \coloneqq i_{{\rm d},k} +j i_{{\rm q},k}$, respectively. We use underlines throughout to indicate complex variables.

\begin{figure}
  \begin{center}
  \includegraphics{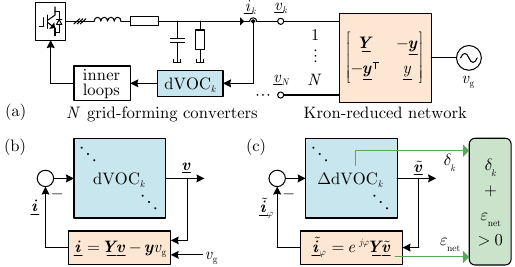}
  \caption{(a) Multi-converter grid-connected systems with dVOC grid-forming control. (b) System model with feedback connection. (c) In incremental form w.r.t to equilibria, the passivity index of the dVOC-controlled node dynamics is denoted as $\delta_k$ while the passivity index of the (rotated) network is denoted as $\varepsilon_{\rm net}$, establishing decentralized stability conditions as $\delta_k + \varepsilon_{\rm net} > 0, \forall k$.}
  \label{fig:block-diagram}
  \end{center}
\end{figure}

\textit{1) Grid-Forming Dynamics:} After transformed into the grid reference frame, the dVOC grid-forming voltage dynamics in complex-voltage coordinates are given as \cite{he2012quanlitative}
\begin{equation}
\label{eq:dvoc}
    \dot{\pha{v}}_k = j\omega_{\Delta} \pha{v}_k + \eta_k e^{j \varphi} \bigl(\tfrac{p_k^{\star} - jq_k^{\star}}{v_k^{\star 2}} \pha{v}_k - \pha i_k \bigr) + \eta_k \alpha_k \tfrac{v_k^{\star 2} - \vert\pha{v}_k\vert^2}{v_k^{\star 2}} \pha{v}_k,
\end{equation}
with gains $\eta_k \in \mathbb{R}_{>0}$, $\alpha_k \in \mathbb{R}_{>0}$, and grid frequency deviation $\omega_{\Delta} \coloneqq \omega_0 - \omega_{\rm g}$. The angle ${\varphi} \in [0,\pi/2]$ is chosen typically as the network impedance angle, and $p_k^{\star}$, $q_k^{\star}$, and $v_k^{\star}$ are the setpoints for active power, reactive power, and voltage amplitude, respectively. In \eqref{eq:dvoc}, the first term $j\omega_{0} \pha{v}_k$ (before rotation) stands for an oscillator at $\omega_0$, the second term synchronizes the relative phases to track the power setpoints via current feedback, and the third term regulates the voltage amplitude.

\textit{2) Network Representation:}
For the set of $N$ converters, we denote $\phavec{v} \coloneqq [\pha v_1,\cdots,\pha v_N]\trans$ and $\phavec{i} \coloneqq [\pha i_1,\cdots,\pha i_N]\trans$. Consider that the network electromagnetic dynamics and the converter inner-loop dynamics are sufficiently fast such that the dVOC dynamics do not interfere with these dynamics. Under this assumption, we consider a static network representation as
\begin{equation}
\label{eq:network}
    \phavec{i} = \phamat{Y}\, \phavec{v} - \phavec{y} v_{\rm g},
\end{equation}
where $v_{\rm g}$ is the grid voltage, and the admittance matrix $\phamat{Y}$ and vector $\phavec{y}$ can be obtained from Kron reduction, where the intermediate nodes in the network are eliminated, see Fig.~\ref{fig:block-diagram}(a). If a virtual impedance is used for current limiting during grid faults \cite{qoria2020current}, it should also be included in \eqref{eq:network}.

\textit{3) System Model With Feedback Connection:}
We synthesize the node dynamics in \eqref{eq:dvoc} and the network equation in \eqref{eq:network} to establish the closed-loop feedback system, as shown in Fig.~\ref{fig:block-diagram}(b). We further define an equivalent system, where the network admittance is rotated by $\varphi$, i.e., the angle employed in dVOC. This equivalent system proves to be more tractable in passivity analysis, and this rotation transformation is also consistent with the typical $p$-$q$ frame transformation in grid-forming feedback controls concerning inductive-resistive networks \cite{de2007voltage}. To do so, we define rotated current as $\pha i_{\varphi,k} \coloneqq e^{j \varphi} \pha i_k$ and $\phavec{i}_{\varphi} \coloneqq e^{j \varphi} \phavec{i}$. The node dynamics in \eqref{eq:dvoc} and the network equation in \eqref{eq:network} are represented equivalently as
\begin{align}
\label{eq:rotated-dvoc}
    \dot{\pha{v}}_k &= \eta_k \bigl[\tfrac{j\omega_{\Delta}}{\eta_k} \pha{v}_k + (\sigma_{\varphi, k}^{\star} + j\rho_{\varphi, k}^{\star}) \pha{v}_k - \pha i_{\varphi,k} + \alpha_k \tfrac{v_k^{\star 2} - \vert\pha{v}_k\vert ^2}{v_k^{\star 2}} \pha{v}_k \bigr] , \\
\label{eq:rotated-network}
    \phavec{i}_{\varphi} &= e^{j \varphi} \phamat{Y}\, \phavec{v} - e^{j \varphi} \phavec{y} v_{\rm g}.
\end{align}
where $\sigma_{\varphi, k}^{\star} + j\rho_{\varphi, k}^{\star} \coloneqq e^{j \varphi} \tfrac{p_k^{\star} - jq_k^{\star}}{v_k^{\star 2}}$. We define steady-state voltages and currents with a subscript $\rm s$ as $\phavec{v}_{\rm s} = [\pha{v}_{{\rm s},1}, \cdots, \pha{v}_{{\rm s},N}]\trans$, $\phavec{i}_{\varphi{\rm s}} = [\pha{i}_{\varphi{\rm s},1}, \cdots, \pha{i}_{\varphi{\rm s},N}]\trans$, and increment variables w.r.t the steady state as $\Tilde{\phavec{v}} \coloneqq \phavec{v} - \phavec{v}_{\rm s}$ and $\Tilde{\phavec{i}}_{\varphi} \coloneqq \phavec{i}_{\varphi} - \phavec{i}_{\varphi{\rm s}}$.

\section{Passivity and Stability Results}
\label{sec:passivity analysis}

We first recall the definition of passivity. The system $\dot {\vect x} = \vect f(\vect x,\vect u)$, $\vect y = \vect h(\vect x,\vect u)$ is said to be output-feedback passive if there exists a positive semidefinite storage function $V(\vect x)$ such that $\vect u\trans \vect y \geq \dot V + \vect y\trans \vect \rho(\vect y)$ for some function $\vect \rho$ \cite[Def. 6.3]{khalil2002nonlinear}. Moreover, it is said to be input-feedforward passive if $\vect u\trans \vect y \geq \dot V + \vect u\trans \vect \varphi(\vect u)$ for some function $\vect \varphi$ \cite[Def. 6.1]{khalil2002nonlinear}. If there exits $\vect \rho(\vect y) = \delta \vect y$ for some $\delta \in \mathbb{R}$ or $\vect \varphi(\vect u) = \varepsilon \vect u$ for some $\varepsilon \in \mathbb{R}$, we refer $\delta$ or $\varepsilon$ to as the passivity index.

\subsection{Passivity Index of Grid-Forming Dynamics}


\begin{proposition}
\label{prop:passivity-dvoc}
The dVOC dynamics in \eqref{eq:rotated-dvoc}, with input $-\pha i_{\varphi,k}$ and output $\pha{v}_k$, are output-feedback passive w.r.t an equilibrium $\pha{v}_{{\rm s},k}$, where the passivity index is given as
\begin{equation}
\label{eq:node-passivity-index}
\delta_k \coloneqq -\Re \bigl\{e^{j\varphi} \tfrac{p_k^{\star} - jq_k^{\star}}{v_k^{\star 2}} \bigr\} - \alpha_k + \alpha_k \tfrac{1}{2 v_k^{\star 2}} \vert \pha v_{{\rm s},k} \vert^2.
\end{equation}
\end{proposition}

\begin{proof}
A positive-definite storage function is defined as
\begin{equation*}
    V_k \coloneqq \tfrac{1}{2 \eta_k} (\phaconj v_k - \phaconj v_{{\rm s},k} ) (\pha v_k - \pha v_{{\rm s},k} ),
\end{equation*}
where $\phaconj v_k$ and $\phaconj v_{{\rm s},k}$ are the \textit{complex conjugate} of $\pha v_k$ and $\pha v_{{\rm s},k}$, respectively (overlines indicate complex conjugate), and $\pha v_{{\rm s},k}$ satisfies the steady-state relationship of \eqref{eq:rotated-dvoc} as
\begin{equation}
\label{eq:steady-relationship}
    \pha i_{\varphi {\rm s},k} = \tfrac{j\omega_{\Delta}}{\eta_k} \pha v_{{\rm s},k} + (\sigma_{\varphi, k}^{\star} + j\rho_{\varphi, k}^{\star}) \pha v_{{\rm s},k} + \alpha_k \tfrac{v_k^{\star 2} - \vert\pha v_{{\rm s},k}\vert ^2}{v_k^{\star 2}} \pha v_{{\rm s},k}.
\end{equation}
The time derivative of $V_k$ along the node dynamics in \eqref{eq:rotated-dvoc} is ${\dot V}_k = \tfrac{1}{2\eta_k} (\phaconj v_k - \phaconj v_{{\rm s},k} ) \dot {\pha v}_k + \tfrac{1}{2\eta_k} \dot{\phaconj v}_k (\pha v_k - \pha v_{{\rm s},k}) = \tfrac{1}{\eta_k} \Re \{(\phaconj v_k - \phaconj v_{{\rm s},k} ) \dot {\pha v}_k \} = \Re \bigl\{(\phaconj v_k - \phaconj v_{{\rm s},k} ) \bigl[\tfrac{j\omega_{\Delta}}{\eta_k} \pha v_k + (\sigma_{\varphi, k}^{\star} + j\rho_{\varphi, k}^{\star})\pha v_k - \pha i_{\varphi {\rm s},k} + \alpha_k \tfrac{v_k^{\star 2} - \vert\pha{v}_k\vert ^2}{v_k^{\star 2}} {\pha v_k} \bigr] \bigr\} + \Re \{(\phaconj v_k - \phaconj v_{{\rm s},k} )(\pha i_{\varphi {\rm s},k} - \pha i_{\varphi,k})\}$. By applying \eqref{eq:steady-relationship}, we obtain that ${\dot V}_k = (\sigma_{\varphi, k}^{\star} + \alpha_k) \vert {\pha v_k - \pha v_{{\rm s},k}} \vert ^2 - \alpha_k \tfrac{1}{v_k^{\star 2}} \Re \bigl\{ (\phaconj v_k - \phaconj v_{{\rm s},k}) (\vert\pha v_k\vert^2 \pha v_k - \vert\pha v_{{\rm s},k}\vert^2 \pha v_{{\rm s},k}) \bigr\} + \Re \{(\phaconj v_k - \phaconj v_{{\rm s},k} )(\pha i_{\varphi {\rm s},k} - \pha i_{\varphi,k})\}$. We bound the second term as $\Re \bigl\{ (\phaconj v_k - \phaconj v_{{\rm s},k}) (\vert\pha v_k\vert^2 \pha v_k - \vert\pha v_{{\rm s},k}\vert^2 \pha v_{{\rm s},k}) \bigr\} \geq \tfrac{1}{2}\vert \pha v_{{\rm s},k} \vert^2 \vert\pha v_k - \pha v_{{\rm s},k} \vert^2 $, which holds due to $(\vect x - \vect y) \trans (\norm{\vect x}^2 \vect x - \norm{\vect y}^2 \vect y) \geq \tfrac{1}{2} \norm{\vect y}^2 \norm{\vect x - \vect y}^2$, $\forall \vect x, \vect y \in \mathbb{R}^2$ \cite[Proposition 6]{he2012quanlitative}, in real variable form. Therefore,
\begin{equation}
\label{eq:bound}
\begin{aligned}
     {\dot V}_k & \leq (\sigma_{\varphi, k}^{\star} + \alpha_k - \alpha_k \tfrac{1}{2 v_k^{\star 2}} \vert \pha v_{{\rm s},k} \vert^2) \vert\pha v_k - \pha v_{{\rm s},k} \vert^2 \\
     & \quad \ + \Re \{(\phaconj v_k - \phaconj v_{{\rm s},k} )(\pha i_{\varphi {\rm s},k} - \pha i_{\varphi,k})\},
\end{aligned}
\end{equation}
which indicates that the node dynamics with incremental input $\pha u_k \coloneqq -( \pha i_{\varphi,k} - \pha i_{\varphi {\rm s},k} )$ and output $\pha y_k \coloneqq (\pha v_k - \pha v_{{\rm s},k})$ is output-feedback passive, as in the standard real-variable form $\vect u \trans \vect y \geq \dot{V} + \delta_k \vect y \trans \vect y$, where $\delta_k = - (\sigma_{\varphi, k}^{\star} + \alpha_k - \alpha_k \tfrac{1}{2 v_k^{\star 2}} \vert\pha v_{{\rm s},k}\vert^2)$ is identified as the node passivity index.
\end{proof}

\begin{remark}
The inequality \eqref{eq:bound} physically means that the power flowing into the converter is greater than the rate of change of the energy storage in the controller together with the excess or shortage of passivity represented by the passivity index. The passivity index is also physically meaningful in the sense that it relates the passivity excess ($\delta_k > 0$) or shortage ($\delta_k < 0$) to the controller setpoints, gains, and steady-state voltage levels. Typically, $\delta_k < 0$ occurs since $\alpha > 0$. High power setpoints are detrimental to passivity while high voltage levels promote passivity, which is in line with engineering experience. Moreover, we can omit the last term of $\delta_k$, making the passivity evaluation independent of the equilibrium point information $\vert \pha v_{{\rm s},k} \vert$, without affecting the rigor of the results (since the last term is always positive).
\end{remark}

The passivity revealed in Proposition \ref{prop:passivity-dvoc} is valid for the \textit{large-signal} model, in contrast to the previous results for small-signal models \cite{dey2023passivity,spanias2019system}. In \cite{kong2023Control}, the dVOC with $\varphi = \pi/2$ is further modified to yield strict passivity w.r.t a prespecific equilibrium point. In general, $e^{j\varphi}$ rotates the current feedback according to the network impedance characteristics. The controller perceives a rotated network \cite{de2007voltage}; thus, the rotated network is the concern of the passivity and stability analysis.

\subsection{Passivity Index of the Network}


\begin{proposition}
\label{prop:network-passivity}
The static network equation in \eqref{eq:rotated-network}, with input $\phavec{v}$ and output $\phavec i_{\varphi}$, is input-feedforward passive w.r.t an equilibrium $\phavec{v}_{\rm s}$ and $\phavec{i}_{\varphi{\rm s}}$, where the passivity index is given as
\begin{equation}
\label{eq:network-passivity-index}
    \varepsilon_{\rm net} \coloneqq \lambda_{\min} \bigl(\Re \{e^{j\varphi} \phamat{Y} \} \bigr) > 0,
\end{equation}
where $\lambda_{\min}$ denotes the minimum eigenvalue.
\end{proposition}

\begin{proof}
In increment variables, it follows from \eqref{eq:rotated-network} that $\Tilde{\phavec{i}}_{\varphi} = e^{j\varphi} \phamat{Y}\, \Tilde{\phavec{v}}$. The proof is completed with $\Re \bigl\{\Tilde{\phavec{v}}\hermconj \Tilde{\phavec{i}}_{\varphi}\bigr\} = \Re \bigl\{\Tilde{\phavec{v}} \hermconj e^{j\varphi} \phamat{Y}\, \Tilde{\phavec{v}} \bigr\} = 1/2(\Tilde{\phavec{v}} \hermconj e^{j\varphi} \phamat{Y}\, \Tilde{\phavec{v}} + \Tilde{\phavec{v}} \hermconj e^{-j\varphi} \phamat{Y}\hermconj\, \Tilde{\phavec{v}}) = \Tilde{\phavec{v}} \hermconj \Re \{e^{j\varphi} \phamat{Y} \}\, \Tilde{\phavec{v}} \ge \lambda_{\min} \bigl(\Re \{e^{j\varphi} \phamat{Y} \}\bigr) \Tilde{\phavec{v}} \hermconj \Tilde{\phavec{v}}$,
where $\Re \bigl\{\Tilde{\phavec{v}}\hermconj \Tilde{\phavec{i}}_{\varphi}\bigr\}$ represents the power flowing into the network.
\end{proof}

\begin{remark}
$\Re \{e^{j\varphi} \phamat{Y} \}$ is real symmetric positive-definite since $\phamat{Y}$ is non-singular, where the row and column associated with the grid node are absent. We relate the network passivity index in \eqref{eq:network-passivity-index} to generalized short-circuit ratio (gSCR), defined as ${\rm gSCR} \coloneqq \lambda_{\min} \bigl(\Re \{e^{j\varphi} \phamat{Y} \} \bigr) = \varepsilon_{\rm net}$, for the concerned rotated inductive-resistive network. In \cite{dong2019small}, an inductive network is concerned, where $\phamat Y$ simplifies to $- j\mat B$, then ${\rm gSCR} = \lambda_{\min} \bigl(\mat{B} \bigr)$ \cite[Def. 2]{dong2019small}, arising from $\varphi = \pi/2$.
\end{remark}

\subsection{Decentralized Stability Conditions}

Theorem~\ref{thm:globally-stable} shows how the passivity properties of dVOC are connected to the transient stability of the closed-loop system.

\begin{theorem}
\label{thm:globally-stable}
Assume that there exists an equilibrium point. The multi-converter grid-connected system in \eqref{eq:rotated-dvoc} and \eqref{eq:rotated-network} is asymptotically stable w.r.t the equilibrium point if it holds that
\begin{equation}
\label{condi:stability}
    \delta_k + \varepsilon_{\rm net} > 0,\ \forall k \in \{1,\cdots,N\}.
\end{equation}
Moreover, if the equilibrium point is unique, the system is globally asymptotically stable.
\end{theorem}

\begin{proof}
We define a composite Lyapunov function as $\nu \coloneqq \sum\nolimits_k V_k$, which is positive definite and radially unbounded w.r.t the equilibrium point. It follows from Propositions \ref{prop:passivity-dvoc} and \ref{prop:network-passivity} that $\dot \nu = \sum\nolimits_k \dot V_k \leq \sum\nolimits_k \Re \{(\phaconj v_k - \phaconj v_{{\rm s},k} )( \pha i_{\varphi {\rm s},k} - \pha i_{\varphi,k} ) \} - \delta_k \vert\pha v_k - \pha v_{{\rm s},k} \vert^2 = -\Re \bigl\{\Tilde{\phavec{v}}\hermconj \Tilde{\phavec{i}}_{\varphi}\bigr\} - \delta_k \Tilde{\phavec{v}} \hermconj \Tilde{\phavec{v}} \leq -\bigl[ \lambda_{\min} \bigl(\Re \{e^{j\varphi} \phamat{Y} \} \bigr) + \delta_k \bigr] \Tilde{\phavec{v}} \hermconj \Tilde{\phavec{v}} = -\bigl( \varepsilon_{\rm net} + \delta_k \bigr) \Tilde{\phavec{v}} \hermconj \Tilde{\phavec{v}}$.
The condition in \eqref{condi:stability} resures that $\dot \nu$ is negative definite w.r.t the equilibrium point. Therefore, the equilibrium point is asymptotically stable, and it is globally asymptotically stable if the equilibrium point is unique.
\end{proof}

\begin{remark}
We do not require each node's dynamics to be strictly passive, in contrast to \cite[Theorem 1]{spanias2019system}; instead, the excess of the network's \textit{input-feedforward passivity} can compensate for the potential shortage of the node's \textit{output-feedback passivity}, mitigating the conservatism of the decentralized stability conditions. Moreover, the unique formulation of dVOC dynamics in $\alpha\beta$ coordinates allows us to exploit the natural passivity of the network. In comparison, the results based on polar coordinates in \cite{yang2020distributed} typically require the passivity of nodes to compensate for the non-passivity of the network coupling.
\end{remark}

\begin{remark}
For a single-converter case, \eqref{condi:stability} reduces to our prior results in \cite{he2012quanlitative}. Additionally, we remark that the decentralized conditions in \eqref{condi:stability} are inapplicable to dVOC-controlled islanded microgrids since the passivity leads the system voltages to converge to zero, where the origin appears as an equilibrium in the case of microgrids \cite{colombino2019global}. To prevent this, the origin is supposed to be unstable, which can be fulfilled with consistent controller setpoints, see \cite{colombino2019global} for further details.
\end{remark}

\section{Validation}

\begin{figure}
  \begin{center}
  \includegraphics[width = 1.0\linewidth]{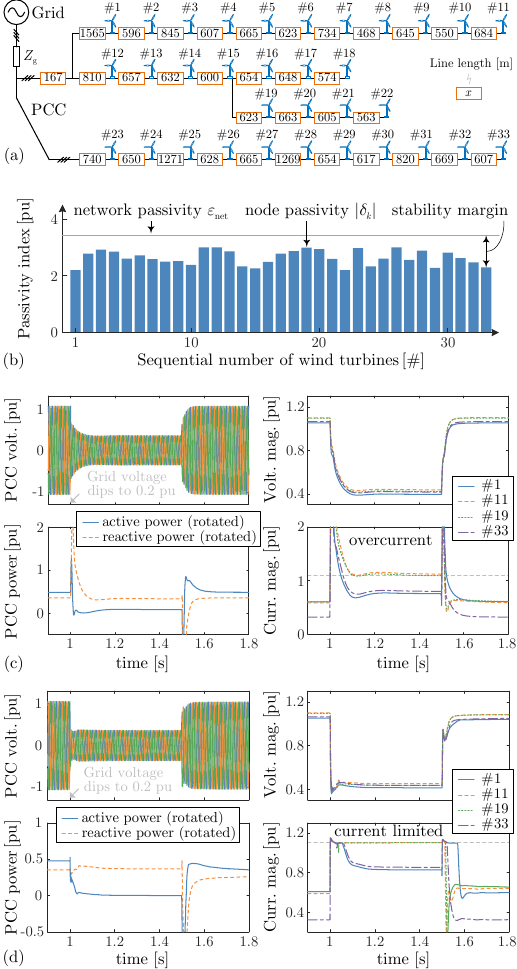}
  \caption{(a) The layout of a real wind power plant. (b) The node and network passivity indices are calculated. The system maintains transient stability under a grid voltage dip, (c) without applying current limiting, or (d) with currents being properly limited throughout by a saturation-informed strategy in \cite{desai2023saturation}.}\vspace{-4mm}
  \label{fig:case-study}
  \end{center}
\end{figure}

We illustrate the calculation of passivity indices and their application in ensuring the transient stability of a wind power plant (WPP). Our case study is based on a real WPP layout depicted in Fig.~\ref{fig:case-study}(a) \cite{li2017practical}, where all wind turbine converters are dVOC-controlled. Power setpoints are randomly selected between $[0, \sqrt{2}]$ pu but scaled down to $1.0$ pu if the apparent power exceeds this threshold. Voltage setpoints are uniformly set to $1.0$ pu. Moreover, we employ $\eta_k = 0.04 \times 100\pi$ rad/s and $\alpha_k = 2$ pu. The grid impedance is $0.05+j0.15$ pu, and the collector line parameters are $0.1153\, \Omega$/${\rm km}$,$1.05\times10^{-3}\, {\rm H}$/${\rm km}$.

The node passivity indices are conservatively estimated by the first two terms of \eqref{eq:node-passivity-index}, where $\delta_k < 0$ happens typically, implying a shortage of node passivity. The network passivity index is computed by \eqref{eq:network-passivity-index}, $\varepsilon_{\rm net} > 0$, contributed by the grid impedance, the collector line impedance, the step-up transformer impedance, and the virtual impedance emulated in converters. Since the equivalent grid impedance usually changes when a grid fault occurs, the changed impedance should be used when evaluating the network passivity. Moreover, the virtual impedance appearing for fault current limiting should also be incorporated into the network representation.

Fig.~\ref{fig:case-study}(b) displays the calculated passivity indices, where the value of $(\varepsilon_{\rm net} - \abs{\delta_k})$ represents the stability margin of each generating unit, $\varepsilon_{\rm net} - \abs{\delta_k} > 0$ for all units satisfying the decentralized stability conditions. The simulation results in Fig.~\ref{fig:case-study}(c) and (d) based on MATLAB/Simulink validate the transient stability of the system under a grid voltage dip, where the complete control dynamics are included, and the converters are average-valued, with fixed DC voltages and sufficiently fast inner-loop dynamics. To address the overcurrent in Fig.~\ref{fig:case-study}(c) while maintaining the grid-forming operation of the converters, we apply a saturation-informed current limiting strategy \cite{desai2023saturation} in Fig.~\ref{fig:case-study}(d). Alternatively, one can use the classical virtual-impedance current limiting strategy \cite{qoria2020current}. Both strategies result in an equivalent circuit with a virtual impedance and an internal voltage source. Our decentralized stability conditions readily apply to the equivalent circuit network with the virtual impedance. We have also validated the stability of the system across different voltage dip depths.


\section{Conclusion}

We have analytically studied the transient stability of multi-converter grid-connected systems employing the grid-forming dVOC control. The passivity of dVOC in large-signal form makes itself noteworthy in stability guarantees. By leveraging the passivity index of the dVOC-controlled node dynamics and the inherent passivity of the network, we develop decentralized stability conditions, serving as a fast and effective tool for controller parameter tuning, transient stability guarantees, and stability margin assessment.

\bibliographystyle{IEEEtran}
\bibliography{IEEEabrv,Bibliography}

\begin{thebibliography}{10}
\providecommand{\url}[1]{#1}
\csname url@rmstyle\endcsname
\providecommand{\newblock}{\relax}
\providecommand{\bibinfo}[2]{#2}
\providecommand\BIBentrySTDinterwordspacing{\spaceskip=0pt\relax}
\providecommand\BIBentryALTinterwordstretchfactor{4}
\providecommand\BIBentryALTinterwordspacing{\spaceskip=\fontdimen2\font plus
\BIBentryALTinterwordstretchfactor\fontdimen3\font minus \fontdimen4\font\relax}
\providecommand\BIBforeignlanguage[2]{{%
\expandafter\ifx\csname l@#1\endcsname\relax
\typeout{** WARNING: IEEEtran.bst: No hyphenation pattern has been}%
\typeout{** loaded for the language `#1'. Using the pattern for}%
\typeout{** the default language instead.}%
\else
\language=\csname l@#1\endcsname
\fi
#2}}
\renewcommand\BIBentryALTinterwordstretchfactor{4}

\bibitem{liu2022stability}
T.~Liu, Y.~Song, L.~Zhu, and D.~J. Hill, ``Stability and control of power grids,'' \emph{Annu. Rev. Control Robot. Auton. Syst.}, vol.~5, no.~1, pp. 689--716, 2022.

\bibitem{li2017practical}
W.~Li, P.~Chao, X.~Liang, J.~Ma, D.~Xu, and X.~Jin, ``A practical equivalent method for {DFIG} wind farms,'' \emph{{IEEE} Trans. Sustain. Energy}, vol.~9, no.~2, pp. 610--620, 2017.

\bibitem{spanias2019system}
C.~Spanias and I.~Lestas, ``A system reference frame approach for stability analysis and control of power grids,'' \emph{{IEEE} Trans. Power Syst.}, vol.~34, no.~2, pp. 1105--1115, 2019.

\bibitem{dey2023passivity}
K.~Dey and A.~M. Kulkarni, ``Passivity-based decentralized criteria for small-signal stability of power systems with converter-interfaced generation,'' \emph{{IEEE} Trans. Power Syst.}, vol.~38, no.~3, pp. 2820--2833, 2023.

\bibitem{yang2020distributed}
P.~Yang, F.~Liu, Z.~Wang, and C.~Shen, ``Distributed stability conditions for power systems with heterogeneous nonlinear bus dynamics,'' \emph{{IEEE} Trans. Power Syst.}, vol.~35, no.~3, pp. 2313--2324, 2020.

\bibitem{khalil2002nonlinear}
H.~K. Khalil, \emph{Nonlinear Systems}, 3rd~ed.\hskip 1em plus 0.5em minus 0.4em\relax Englewood Cliffs, NJ, USA: Prentice-Hall, 2002.

\bibitem{he2012quanlitative}
\BIBentryALTinterwordspacing
X.~He, L.~Huang, I.~Subotić, V.~Häberle, and F.~Dörfler, ``Quantitative stability conditions for grid-forming converters with complex droop control,'' 2023, submitted to \textit{IEEE Trans. Power Electron.} [Online]. Available: \url{https://arxiv.org/pdf/2310.09933.pdf}
\BIBentrySTDinterwordspacing

\bibitem{qoria2020current}
T.~Qoria, F.~Gruson, F.~Colas, X.~Kestelyn, and X.~Guillaud, ``Current limiting algorithms and transient stability analysis of grid-forming {VSCs},'' \emph{Electr. Power Syst. Res.}, vol. 189, p. 106726, 2020.

\bibitem{de2007voltage}
K.~De~Brabandere, B.~Bolsens, J.~Van~den Keybus, A.~Woyte, J.~Driesen, and R.~Belmans, ``A voltage and frequency droop control method for parallel inverters,'' \emph{{IEEE} Trans. Power Electron.}, vol.~22, no.~4, pp. 1107--1115, 2007.

\bibitem{kong2023Control}
L.~Kong, Y.~Xue, L.~Qiao, and F.~Wang, ``Control design of passive grid-forming inverters in port-hamiltonian framework,'' \emph{{IEEE} Trans. Power Electron.}, pp. 1--15, 2023.

\bibitem{dong2019small}
W.~Dong, H.~Xin, D.~Wu, and L.~Huang, ``Small signal stability analysis of multi-infeed power electronic systems based on grid strength assessment,'' \emph{{IEEE} Trans. Power Syst.}, vol.~34, no.~2, pp. 1393--1403, 2019.

\bibitem{colombino2019global}
M.~Colombino, D.~Groß, J.-S. Brouillon, and F.~Dörfler, ``Global phase and magnitude synchronization of coupled oscillators with application to the control of grid-forming power inverters,'' \emph{{IEEE} Trans. Autom. Control}, vol.~64, no.~11, pp. 4496--4511, 2019.

\bibitem{desai2023saturation}
M.~A. Desai, X.~He, L.~Huang, and F.~D\"{o}rfler, ``Saturation-informed current-limiting control for grid-forming converters,'' 2023, submitted to \textit{Electr. Power Syst. Res.} and patent pending.

\end{thebibliography}

\end{document}